\documentclass[journal]{ieeecolor}
\usepackage{lcsys}
\usepackage{cite}
\usepackage{amsmath,amssymb,amsfonts}
\usepackage{algorithmic}
\usepackage{graphicx}
\usepackage{textcomp}
\usepackage{epstopdf}
\usepackage{subcaption}
\usepackage{tikz}
\usepackage{arydshln}
\newtheorem{theorem}{Theorem}
\newtheorem{lemma}[theorem]{Lemma}
\newtheorem{proposition}[theorem]{Proposition}
\newtheorem{corollary}[theorem]{Corollary}
\newtheorem{definition}[theorem]{Definition}
\newtheorem{remark}[theorem]{Remark}
\newtheorem{assumption}[theorem]{Assumption}

\newtheorem{problem}[theorem]{Problem}
\newtheorem{example}[theorem]{Example}

\newcommand{\calA}{\ensuremath{\mathcal{A}}}
\newcommand{\calB}{\ensuremath{\mathcal{B}}}
\newcommand{\calC}{\ensuremath{\mathcal{C}}}

\newcommand{\calH}{\ensuremath{\mathcal{H}}}
\newcommand{\calI}{\ensuremath{\mathcal{I}}}

\newcommand{\calM}{\ensuremath{\mathcal{M}}}

\newcommand{\calP}{\ensuremath{\mathcal{P}}}

\newcommand{\calW}{\ensuremath{\mathcal{W}}}

\newcommand{\EP}{\hspace*{\fill}~\mbox{\rule[0pt]{1.3ex}{1.3ex}}\par\endtrivlist\unskip}
\usepackage{enumerate}

\makeatletter
\let\NAT@parse\undefined
\makeatother
\usepackage{hyperref}

\def\BibTeX{{\rm B\kern-.05em{\sc i\kern-.025em b}\kern-.08em
    T\kern-.1667em\lower.7ex\hbox{E}\kern-.125emX}}
\begin{document}
\title{Strong Structural Controllability of Structured Networks with MIMO node systems}
\author{
Yanting Ni, Xuyang Lou, Junjie Jiao, and Jiajia Jia$^{*}$
\thanks{Y. Ni, X. Lou and J. Jia are with the Jiangnan University,  Wuxi, 214122, P.R. China, (e-mail: niyant@163.com, Louxy@jiangnan.edu.cn, j.jia@jiangnan.edu.cn); Junjie Jiao is with the Chair of Information Oriented Control, TUM School of Computation, Information and Technology, Technical University of Munich, 81669 Munich, Germany; \texttt{ junjie.jiao@tum.de} .(Corresponding author: Jiajia Jia)}
}


\maketitle
\thispagestyle{empty}
\begin{abstract}
The article addresses the problem of strong structural controllability of structured networks with multi-input multi-output (MIMO) node systems.
The authors first present necessary and sufficient conditions for strong structural controllability, which involve both algebraic and graph-theoretic aspects.
These conditions are computationally expensive, especially for large-scale networks with high-dimensional state spaces.
To overcome this computational complexity, we propose a necessary algebraic condition from a node system's perspective and a graph-theoretic condition from a network topology's perspective.
The latter condition is derived from the structured interconnection laws and employs a new color change rule, namely {\em weakly color change rule} introduced in this paper.
Overall, this article contributes to the study of strong structural controllability in structured networks with MIMO node systems, providing both theoretical and practical insights for their analysis and design.
\end{abstract}


\section{Introduction}
\label{sec:introduction}
\IEEEPARstart{I}{n} the omni-networking world today, more and more challenging theoretical problems are encountered in the analysis and synthesis of networks. In the past decades, the problem of controllability of networks has been investigated with much interest and some efficient criteria have been established \cite{Trumpf2018controllability,Xiang2020Controllability,ZHOU2015On,hao2019new,JLY2015}.
However, in many real-world networks, no matter the parameter values of systems or the strengths of interconnections between systems cannot be accurately obtained except for the zeroes that mean the absence of connections.
Typical examples include Internet, power grids, biological networks, transportation networks, and so on.
For this reason, based on Lin's weak structural controllability theory \cite{lin1974structural}, analytical tools were firstly provided to analyze weakly structurally controllable for a directed network in \cite{liu2011controllability}.

Following this, weakly structural controllability of networks has been widely studied from algebraic and graphical perspectives, and numerous efficient criteria have been established, see e.g., \cite{monshizadeh2014zero,commault2019structural,mu2022structural,Zhang2019Structural}.
We call a network \emph{weakly} structurally controllable if for almost all numerical realizations such that the associated network is controllable.
It should be pointed out that although uncontrollable networks with the same structure as a weakly structurally controllable network are atypical, in some cases the existence of such networks is not permitted (see \cite{Mayeda1979Strong}). Consequently, to relax the above-mentioned restriction, \emph{strong} structural controllability of networks has become a focal subject in recent years. A structured network is called (\emph{strongly} structurally) controllable if for all numerical realization such that the associated network is controllable.

Within the paradiam of strong structural controllability, the major part of related work has been done, where most, if not all, results on strong structural controllability are derived under the assumption that all node systems are single integrators, see e.g., \cite{Chapman2013On,menara2017number,trefois2015zero}. And in \cite{jia2020unifying}, the concept of structured networks was proposed, which is formed by interconnecting structured node systems and external control inputs via structured interconnection laws, and provided a unifying framework for strong structural controllability of structured networks.

Motivated by the fact that in most real dynamical networks, the node systems might have higher dimensions. Hence, \cite{jia2021scalable} further studied strong structural controllability of structured networks with single-input single-output (SISO) node systems, and established algebraic and graphical conditions for strong structural controllability of structured networks. When it comes to strong structural controllability of structured networks with multi-input multi-output (MIMO) node systems, to the best of our knowledge, this has not been studied before.

Motivated by the above-mentioned discussions, the strong structural controllability of structured networks with MIMO node systems is investigated in this article.
The main contributions of this article are the following.
\begin{enumerate}[1)]
  \item
  We study strong structural controllability of structured networks with general heterogeneous MIMO node systems, which generalizes the results of \cite{jia2020unifying} and \cite{jia2021scalable}.
  \item
  For large-scale structured networks, these criteria generally fail to work because of extremely high dimensions, which makes the computation very expensive. To relax the above-mentioned limitation, we provide conditions for strong structural controllability of large-scale structured networks from two perspectives, i.e. node systems and network topology.
  \item
  We introduce a new {\it weakly color change rule} to verify the graph-theoretic condition from the network topology.
\end{enumerate}

The outline of this article is as follows. In Section \ref{sec:preliminaries}, we formulate the problem considered in this article, and present some preliminary results.
In Section \ref{sec:3}, we provide algebraic and graph-theoretic conditions for strong structural controllability of structured networks with MIMO node systems. In Section \ref{sec:necessary} which states the main results in this article, we establish necessary conditions for strong structural controllability of large-scale structured networks from the perspective of node systems and the underlying network topology extracted from the structured interconnection laws. Furthermore, we introduce a new color change rule to verify the necessary condition. Finally, Section \ref{sec:condition} concludes this article.

The following notations and symbols are adopted.
Let $\mathbb{R}$ and $\mathbb{R}^n$ denote the fields of real numbers and the spaces of $n$-dimensional real vectors, respectively.
Likewise, the space of $n \times m$ real matrices is denoted by $\mathbb{R}^{n \times m}$.
For a given $n \times m$ matrix $A$, the entry in the $i$th row and $j$th column is denoted by $A_{ij}$.
For a given $n \times m$ block matrix $A$, the block on the $i$th row and $j$th column is denoted by $A^{(ij)}$.
For a given set of matrices $\{ A^{(1)}, \ldots, A^{(n)}\}$, $ \text{diag}(A^{(1)}, \cdots, A^{(n)})$ denotes the $n \times n$ blocked diagonal matrix, and $ \text{col}(A^{(1)}, \cdots, A^{(n)})$ denotes the vector stacked by $A^{(1)},\ldots,A^{(n)}$ in a column.
By a pattern matrix, which plays an important role throughout this article, we mean a matrix with entries in the set of symbols $\{ 0,\ast,? \}$.
The set of all $p \times q$ pattern matrices will be denoted by $\{ 0,\ast,? \} ^{p \times q}$.
For a given pattern matrix $\mathcal{M} \in \{ 0,\ast,? \} ^{p \times q}$, the pattern class of $\mathcal{M}$ is defined as the subset of $\mathbb{R}^{p \times q}$ given by
\begin{equation*}
\begin{aligned}
    \mathcal{P}(\mathcal{M}) = \{ M \in \mathbb{R}^{p \times q} \mid & M_{ij} = 0 \quad \text{if} \quad \mathcal{M}_{ij} = 0,\\
    & M_{ij} \neq 0 \quad \text{if} \quad \mathcal{M}_{ij} = \ast \}.
\end{aligned}
\end{equation*}
This means that for a matrix $M \in \mathcal{P}(\mathcal{M})$, the entry $M_{ij}$ is equal to the real number $0$ if $\mathcal{M}_{ij} = 0$, a nonzero real number if $\mathcal{M}_{ij} = \ast$, and an arbitrary real number if $\mathcal{M}_{ij} = ?$.
Specially, we denote  by $\textbf{0}$ the pattern matrix with all zero entries of appropriate dimensions, and by $\mathcal{I}$ the square pattern matrix of appropriate dimensions, in which the diagonal entries are equal to $\ast$ and others are $0$.

\section{PROBLEM FORMULATION AND PRELIMINARIES}\label{sec:2}
\label{sec:preliminaries}

\subsection{Problem Formulation}

Consider a structured network composed of $N$ structured systems.
The $k$th structured systems $(\mathcal{A}^{(k)}, \mathcal{B}^{(k)}, \mathcal{C}^{(k)})$, called the {\it structured node system} at node $k$, has the following dynamics
\begin{equation}\label{eq:nodesystems}
    \begin{cases}
       \dot{x}^{(k)} &= A^{(k)} x^{(k)} +  B^{(k)} v^{(k)},\\
    y^{(k)} &= C^{(k)} x^{(k)},
    \end{cases}
\end{equation}
where $A^{(k)} \in \calP(\mathcal{A}^{(k)})$, $B^{(k)} \in \calP(\mathcal{B}^{(k)})$, and $C^{(k)} \in \calP(\mathcal{C}^{(k)})$ have dimensions $n_k \times n_k$, $n_k \times r_k$, $p_k \times n_k$, respectively.
In addition, these $N$ node systems are interconnected via structured interconnection law given by block pattern matrices
\begin{equation*}
   \begin{scriptsize}\mathcal{W} = \begin{bmatrix}
                   \mathcal{W}^{(11)} & \cdots & \mathcal{W}^{(1N)} \\
                   \vdots & \ddots & \vdots \\
                   \mathcal{W}^{(N1)} & \cdots & \mathcal{W}^{(NN)}
                 \end{bmatrix},~
   \mathcal{H} = \begin{bmatrix}
                   \mathcal{H}^{(11)} & \cdots & \mathcal{H}^{(1m)} \\
                   \vdots & \ddots & \vdots \\
                   \mathcal{H}^{(N1)} & \cdots & \mathcal{H}^{(Nm)}
                 \end{bmatrix}\end{scriptsize}
\end{equation*}
where $\mathcal{W}$ and $\mathcal{H}$ have dimensions $r \times p$ and $r \times m$. Here $r := \sum^{N}_{k = 1} r_k$ and $p := \sum^{N}_{k = 1} p_k$.
More explicitly, the input signal $v^{(k)}$ injected to $(\mathcal{A}^{(k)}, \mathcal{B}^{(k)}, \mathcal{C}^{(k)})$ may contain both node system interactions and external control input,
\begin{equation*}
v^{(k)} = \sum\limits_{j=1}^{N} W^{(kj)} y^{(j)} + \sum\limits_{i=1}^{m} H^{(ki)} u_i,
\end{equation*}
where $W^{(kj)} \in \mathcal{P}(\mathcal{W}^{(kj)})$ describes the interconnection from node $j$ to node $k$, and $H^{(ki)} \in \mathcal{P}(\mathcal{H}^{(ki)})$ describes the interconnection from external input $u_i$ to node $k$ with $i=1,\ldots,m$.
 By introducing the block diagonal matrices
\begin{equation*} \label{eq:ABC}
    \begin{aligned}
    A &= \text{diag}(A^{(1)},\ldots,A^{(N)}),\\
    B &= \text{diag}(B^{(1)},\ldots,B^{(N)}),\\
    C &= \text{diag}(C^{(1)},\ldots,C^{(N)}),
    \end{aligned}
\end{equation*}
and the block matrices
\begin{scriptsize}
\begin{equation*} \label{eq:WH}
 W = \begin{bmatrix}
       W^{(11)} & \ldots & W^{(1N)} \\
       \vdots & \ddots & \vdots \\
       W^{(N1)} & \ldots & W^{(NN)}
       \end{bmatrix},
  H = \begin{bmatrix}
       H^{(11)} & \ldots & H^{(1m)} \\
        \vdots & \ddots & \vdots \\
        H^{(N1)} & \ldots & H^{(Nm)}
        \end{bmatrix},
\end{equation*}
\end{scriptsize}
the system \eqref{eq:nodesystems} is rewritten in a compact form as
\begin{equation} \label{eq:comnetwork}
    \dot{x} = (A + BWC)x + BHu,
\end{equation}
where $x = \text{col}(x_1,\ldots,x_N)$ and $u = \text{col}(u_1,\ldots,u_m)$.
Here, $x \in \mathbb{R}^{n}$ with $n := \sum_{k=1}^{N} n_k$.

Now introduce the block pattern matrices
\begin{equation} \label{eq:patternABC}
    \begin{aligned}
    \mathcal{A} &= \text{diag}(\mathcal{A}^{(1)},\ldots,\mathcal{A}^{(N)}),\\
    \mathcal{B} &= \text{diag}(\mathcal{B}^{(1)},\ldots,\mathcal{B}^{(N)}),\\
    \mathcal{C} &= \text{diag}(\mathcal{C}^{(1)},\ldots,\mathcal{C}^{(N)}).
    \end{aligned}
\end{equation}
It is easy to see that structured networks are a collection of systems \eqref{eq:comnetwork}, where $A \in \mathcal{P}(\mathcal{A}),B \in \mathcal{P}(\mathcal{B}),C \in \mathcal{P}(\mathcal{C}),W \in \mathcal{P}(\mathcal{W}),H \in \mathcal{P}(\mathcal{H})$.
For the sake of simplicity, the structured network will be denoted by $(\mathcal{A},\mathcal{B},\mathcal{C},\mathcal{W},\mathcal{H})$.
The structured network $(\mathcal{A},\mathcal{B},\mathcal{C},\mathcal{W},\mathcal{H})$ is called {\em strongly structurally controllable} (shortly controllable) if \eqref{eq:comnetwork} is  controllable for all $A \in \mathcal{P}(\mathcal{A}),B \in \mathcal{P}(\mathcal{B}),C \in \mathcal{P}(\mathcal{C}),W \in \mathcal{P}(\mathcal{W}),H \in \mathcal{P}(\mathcal{H})$.

Strong structural controllability of such a structured network was studied assuming that node systems are single integrators or SISO in \cite{jia2020unifying} and \cite{jia2021scalable}, respectively.
However, in representing node systems of networks, capturing the system simply by single integrators or SISO is not always possible; thus, a more general MIMO node system is required.
Since the introduction of MIMO node dynamics makes the whole network more complicated, the results in the existing literature need to be revised.
Motivated by the limitations above, the research problem considered in this article is formulated as follows.
\begin{problem}\label{promblem1}
Find conditions under which the structured network with MIMO node systems is controllable.
\end{problem}

\subsection{Properties and operations of pattern matrices}

To begin with,  we review definitions of sums and products of pattern matrices introduced in \cite{SWCT2021}.
The definition of addition and multiplication of the symbols $0,\ast$ and $?$ is presented in the following Table \ref{ta:results}.

\begin{table}[h!]
	\centering
	\caption{Addition and multiplication within the set $\{0,\ast,?\}$.}
	\begin{tabular}{c|ccc}
		$+$ & $0$ & $\ast$ & $?$  \\  \hline
		$0\rule{0pt}{2.2ex}$ & $0$ & $\ast$ & ${?}$ \\
		$\ast$& $\ast$ & $?$ & $?$  \\
		$?$& $?$ & $?$ & $?$
	\end{tabular}
	\qquad
	\begin{tabular}{c|ccc}
		$\boldsymbol{\cdot}$  & $0$ & $\ast$ & $?$ \\  \hline
		$0\rule{0pt}{2.2ex}$ &  $0$ & $0$ & $0$ \\
		$\ast$& $0$ & $\ast$ & $?$ \\
		$?$&  $0$ & $?$ & $?$
	\end{tabular}
	\label{ta:results}
\end{table}
Then, the definition of addition of pattern matrices, based on the operations defined in this table, is given as follows.
\begin{definition}
\label{d:mPa}
Consider pattern matrices  $\mathcal{M}, \mathcal{N} \in \{0,\ast,?\}^{p \times q}$.
Then the sum $\mathcal{M} + \mathcal{N} \in \{0,\ast,?\}^{p \times q}$ is defined as
\begin{equation*}
(\mathcal{M} + \mathcal{N})_{ij}:=  \mathcal{M}_{ij} + \mathcal{N}_{ij}.
\end{equation*}
\end{definition}
Furthermore, $\mathcal{P}(\calM) + \mathcal{P}(\mathcal{N})$ is defined as the usual Minkowski sum of sets, which means that
$$
\mathcal{P}(\mathcal{M}) + \mathcal{P}(\mathcal{N}) := \{M+N \mid M \in \mathcal{P}(\mathcal{M}) \mbox{ and }  N \in \mathcal{P}(\mathcal{N})\}.
$$
We now have the following proposition.
\begin{proposition}
\label{p:1} \cite[Proposition 1]{SWCT2021}
For pattern matrices $\calM$ and $\mathcal{N}$ of the same dimensions,
$\calP(\calM) + \mathcal{P}(\mathcal{N}) = \calP(\calM + \mathcal{N})$.
\end{proposition}
Next, we review the definition of multiplication of pattern matrices.
\begin{definition}
\label{d:mPm}
Consider pattern matrices $\calM \in \{0,\ast,?\}^{p \times q}$ and $\mathcal{N}  \in \{0,\ast,?\}^{q \times s}$.
Then the product $\calM \mathcal{N}  \in \{0,\ast,?\}^{p \times s}$ is defined by
\begin{equation*}
(\calM \mathcal{N})_{ij} :=  \sum_{\ell = 1}^{q} \calM_{i\ell} \boldsymbol{\cdot} \mathcal{N}_{\ell j}.
\end{equation*}
\end{definition}
We define
$
\calP(\calM)\calP(\mathcal{N}) := \{MN \mid M \in \calP(\calM) \mbox{ and }  N \in \calP(\mathcal{N})\}.
$
Unfortunately, for general pattern matrices $\mathcal{M}$ and $\mathcal{N}$, the equality $\calP(\calM)\calP(\mathcal{N}) = \calP(\calM \mathcal{N})$ \emph{does not hold} \cite[Example 1]{SWCT2021}. Nonetheless, if $\mathcal{M}$ and $\mathcal{N}$ have a special structure as we present next, such an equality can be derived.
\begin{lemma}\cite[Lemma 5]{jia2021scalable}
\label{l:PMP}
Consider two pattern matrices $\calM \in \{0,\ast,?\}^{p \times q}$ and $\mathcal{N} \in \{0,\ast,?\}^{q \times r}$.
Then, the equality $$ \calP(\calM)\calP(\mathcal{N}) = \calP(\calM \mathcal{N})$$
holds if at least one of the following two conditions.
\begin{enumerate}
\item Each row of $\mathcal{N}$ has exactly one entry equal to $\ast$ and the remaining entries are zero.
\item Each column of $\calM$ has exactly one entry equal to $\ast$ and the remaining entries are zero.
\end{enumerate}
\end{lemma}

A pattern matrix $\mathcal{M} \in \{0,\ast,?\}^{p \times q}$ with $p \le q$ is said to have {\em full row rank} if $M$ has full row rank for every $M \in \mathcal{P}(\mathcal{M})$.
In order to introduce graph theoretic conditions for full rank properties of pattern matrices, we define the digraph associated with $\mathcal{M} \in \{0,\ast,?\}^{p \times q}$  with $p \le q$ as $G(\mathcal{M}) = (V,E)$ as follows.
 Take a vertex set $V = \{1,2, \ldots, q\}$, and define an edge set $E \subseteq V \times V$ such that $(j,i) \in E$ if and only if $\mathcal{M}_{ij} =\ast$ or $\mathcal{M}_{ij} =?$.
 We call $j$ an out-neighbor of $i$ if $(i,j) \in E$.
In addition, to distinguish between $\ast$ and $?$ entries in $\mathcal{M}$, we define two subsets $E_\ast$ and $E_?$ of the edge set $E$ as follows: $(j,i) \in E_\ast$ if and only if $\mathcal{M}_{ij} = \ast$ and $(j,i) \in E_?$ if and only if $\mathcal{M}_{ij} = ?$.

Consider the following color change rule introduced in \cite{jia2020unifying}.
\begin{enumerate}
\item Initially, color all vertices in $V$ white.
\item If a vertex $i$ has exactly one white out-neighbor $j$ and $(i,j) \in E_{\ast}$, then change the color of $j$ to black.
\item repeat step 2 until no more color changes are possible.
\end{enumerate}
The derived set $\mathcal{D}(\calM)$ of $G(\calM)$ is defined as the set of all black nodes obtained by applying the above procedure to $G(\calM)$.
In the special case that $\mathcal{D}(\calM) = \{1,2,\ldots,p\}$, we call the graph $G(\calM)$ {\it colorable}.
It has been shown in \cite[Theorem 10]{jia2020unifying} that $\calM$ has full row rank if and only if $G(\calM)$ is colorable, i.e., $\mathcal{D}(\calM) = \{1,2,\ldots,p\}$.

\subsection{Conditions for controllability of $(\mathcal{A},\mathcal{B})$}
For given pattern matrices $\mathcal{A} \in \{ 0,\ast,? \}^{n \times n}$, $\mathcal{B} \in \{ 0,\ast,? \}^{n \times r}$ and $\mathcal{C} \in \{ 0,\ast,? \}^{p \times n}$, the structured
system associated with these pattern matrices is defined as the family of LTI systems
\begin{align}
    \dot{x} &=Ax+Bu,\label{LTI1}\\
        y &=Cx,
\end{align}
where $A \in \mathcal{P}(\mathcal{A}),B \in \mathcal{P}(\mathcal{B})$ and $C \in \mathcal{P}(\mathcal{C})$.
For the sake of simplicity, we denote the above structured system by $(\mathcal{A},\mathcal{B},\mathcal{C})$.
Similarly, we will denote the family of systems \eqref{LTI1} by $(\mathcal{A},\mathcal{B})$.
If \eqref{LTI1} is controllable for all $A \in \mathcal{P}(\mathcal{A})$ and $B \in \mathcal{P}(\mathcal{B})$, the structured system $(\mathcal{A},\mathcal{B})$ is called {\it controllable}.

Recently,  \cite{jia2020unifying} linked controllability of $(\mathcal{A},\mathcal{B})$ to full row rank properties of two pattern matrices  $\begin{bmatrix}
\calA & \calB
\end{bmatrix}$ and $\begin{bmatrix}
\calA + \calI & \calB
\end{bmatrix}$.
More explicitly, both algebraic and graph theoretic necessary and sufficient conditions under which a given structured network is controllable have been stated as follows.
\begin{proposition}\cite[Theorem 6 \& 11]{jia2020unifying} \label{p:T4SS}
The following statements are equivalent:
\begin{enumerate}
\item a structured system $(\calA,\calB,\calC)$ given by \eqref{LTI1} is controllable;
\item both the pattern matrices  $\begin{bmatrix}
\calA & \calB
\end{bmatrix}$ and $\begin{bmatrix}
\calA+\calI & \calB
\end{bmatrix}$ have full row rank.
\item both the graphs $G(\begin{bmatrix}
\calA & \calB
\end{bmatrix})$ and $G(\begin{bmatrix}
\calA+\calI & \calB
\end{bmatrix})$ are colorable.
\end{enumerate}
\end{proposition}

\section{CONDITIONS FOR CONTROLLABILITY OF STRUCTURED NETWORKS}\label{sec:3}

In this section, we will analyze conditions under which the structured network with MIMO node systems is controllable. Since single integrators and SISO node systems are special forms of MIMO nodes, we will try to generalize the results of \cite{jia2020unifying} and \cite{jia2021scalable} to strong structural controllability of structured networks with general heterogeneous MIMO node systems.

Before presenting our results, without loss of generality, we first make the following simplifying assumption that will be in place throughout the article.

\begin{assumption}\label{as:1}
For all $k \in \{1,2 \ldots,N\}$,  each input (output) of the node system $(\calA^{(k)},\calB^{(k)},\calC^{(k)})$ can be injected to (affected by) exactly one state.
\end{assumption}

By Assumption \ref{as:1}, for all $k \in \{1,2 \ldots,N\}$, every entries of $\mathcal{B}^{(k)}$ and $(\mathcal{C}^{(k)})^\top$are equal to $0$ except for exactly one equal to $\ast$ for each column.
This special structure will allow us to apply Lemma \ref{l:PMP}.
Note that Assumption \ref{as:1} is not really restriction of Problem \ref{promblem1}.
We are ready to provide the first result in this paper as follows.
\begin{theorem}\label{th:C4SN}
The following statements are equivalent:
\begin{enumerate}[(i)]
\item the structured network $(\mathcal{A},\mathcal{B},\mathcal{C},\mathcal{W},\mathcal{H})$ is controllable \label{st:1};
\item the structured system $(\calA + \calB \calW \calC, \calB \calH)$ is controllable \label{st:2};
\item both pattern matrices $\begin{bmatrix}
   \mathcal{A} + \mathcal{B}\mathcal{W}\mathcal{C} & \mathcal{B}\mathcal{H}
 \end{bmatrix}$ and $\begin{bmatrix}
   \calA+\calI + \mathcal{B}\mathcal{W}\mathcal{C} & \mathcal{B}\mathcal{H}
 \end{bmatrix}$ have full row rank\label{st:3};
 \item both graphs $G(\begin{bmatrix}
     \mathcal{A} + \mathcal{B}\mathcal{W}\mathcal{C} & \mathcal{B}\mathcal{H}
   \end{bmatrix})$ and $G(\begin{bmatrix}
       \calA+\calI + \mathcal{B}\mathcal{W}\mathcal{C} & \mathcal{B}\mathcal{H}
     \end{bmatrix})$ are colorable \label{st:4};
\end{enumerate}
\end{theorem}

\begin{proof}
The equivalence of statements \eqref{st:2}, \eqref{st:3} and \eqref{st:4} are followed immediately from Proposition \ref{p:T4SS}.
Therefore, we only need to show that the statement \eqref{st:1} is equivalent to the statement \eqref{st:2} as follows.
 Firstly, by Proposition \ref{p:1}, we have that
  $$\calP(\calA + \calB\calW\calC) = \calP(\calA) + \calP(\calB \calW \calC).$$
  In addition, by Assumption \ref{as:1} and Lemma \ref{l:PMP}, we have that
 $$\calP(\calA + \calB\calW\calC) = \calP(\calA) + \calP(\calB) \calP(\calW) \calP(\calC)$$
 $$\calP(\calB\calH) = \calP(\calB) \calP(\calH),$$
 which means that statements \eqref{st:1} and \eqref{st:2} are equivalent, and this completes the proof.
\end{proof}

Next, we present the following example to illustrate Theorem \ref{th:C4SN}.
\begin{example}\label{ex:1}
Consider a structured network consisting of 3 two-input two-output node systems given by
\begin{equation*}
\begin{scriptsize}
\begin{aligned}
  \mathcal{A}^{(1)} & = \begin{bmatrix}
                     \ast & 0 & 0 & 0 \\
                     0 & ? & 0 & 0 \\
                     ? & \ast & \ast & 0 \\
                     \ast & 0 & 0 & ?
                      \end{bmatrix},
  \mathcal{A}^{(2)} = \begin{bmatrix}
                        ? & 0 & \ast & 0 \\
                        0 & \ast & 0 & \ast \\
                        0 & \ast & \ast & 0 \\
                        \ast & 0 & 0 & ?
                      \end{bmatrix},\\
   \mathcal{A}^{(3)} & = \begin{bmatrix}
                        \ast & 0 & 0 & 0 \\
                        0 & 0 & \ast & 0 \\
                        0 & 0 & ? & \ast \\
                        \ast & 0 & \ast & \ast
                      \end{bmatrix},
  \mathcal{B}^{(1)} = \mathcal{B}^{(2)} = \mathcal{B}^{(3)} = \begin{bmatrix}
                        \ast & 0 \\
                        0 & \ast \\
                        0 & 0\\
                        0 & 0
                      \end{bmatrix},\\
  \mathcal{C}^{(1)} &= \mathcal{C}^{(2)} = \mathcal{C}^{(3)} = \begin{bmatrix}
                        0 & 0 & \ast & 0 \\
                        0 & 0 & 0 & \ast
                      \end{bmatrix},
\end{aligned}
\end{scriptsize}
\end{equation*}
and an external input vector $u \in \mathbb{R}^2$.
The structured network is formed by interconnecting these node systems through the structured interconnection law defined by the following pattern matrices
\begin{equation*}
  \mathcal{W} := \begin{scriptsize}
 \left[\begin{array}{cc:cc:cc}
                          0 & 0 & 0 & 0 & 0 & 0 \\
                          0 & 0 & 0 & 0 & 0 & 0 \\
                          \hdashline
                          \ast & 0 & 0 & 0 & 0 & 0 \\
                          ? & \ast & 0 & 0 & 0 & 0 \\
                          \hdashline
                          0 & 0 & \ast & 0 & 0 & 0 \\
                          0 & 0 & 0 & ? & 0 & 0
                        \end{array}\right]
  \end{scriptsize} ~~\text{and}~~
  \mathcal{H} := \begin{scriptsize}
  \left[\begin{array}{c:c}
                           \ast & 0  \\
                           0 & \ast  \\
                           \hdashline
                           0 & 0   \\
                           0 & 0  \\
                           \hdashline
                           0 & 0  \\
                           0 & 0
                         \end{array}\right]
  \end{scriptsize}.
\end{equation*}
By applying the color change rule in Section \ref{sec:2}, it turns out that both the graph $G(\begin{bmatrix}
\calA+\calB\calW\calC & \calB\calH
\end{bmatrix})$ and $G(\begin{bmatrix}
\calA + \calI+\calB\calW\calC & \calB\calH
\end{bmatrix})$ depicted in Figure 1 is colorable, and hence the network $(\calA,\calB,\calC,\calW,\calH)$ is controllable.
\vspace{-5pt}
\begin{figure}[!h]
\centering 
\includegraphics[scale=0.5]{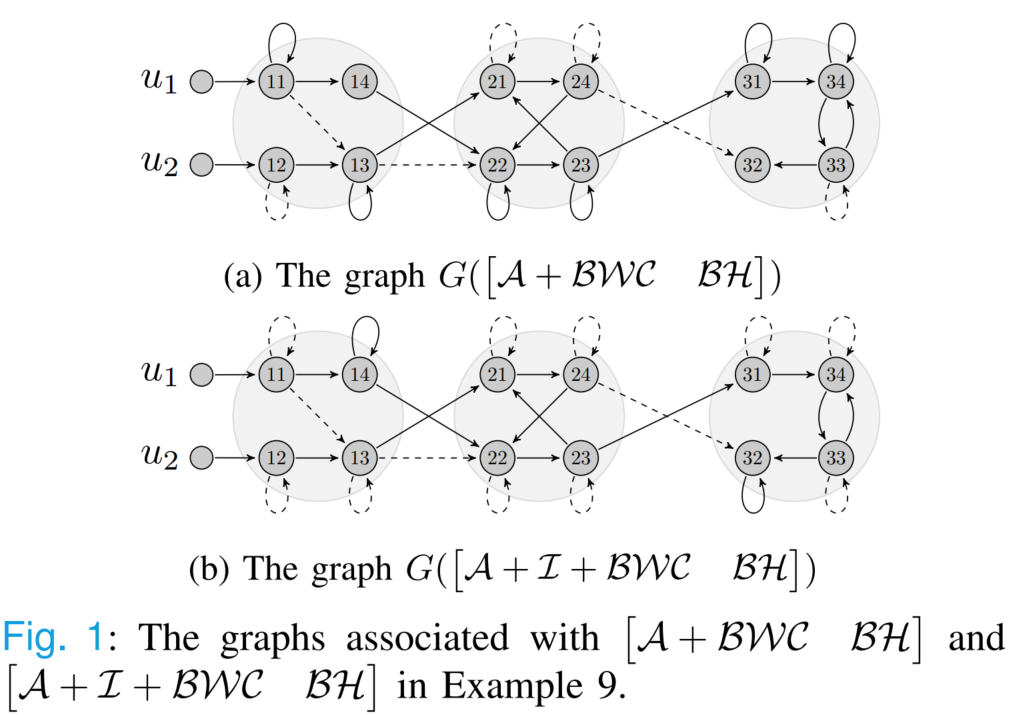}

\label{1}
\end{figure}
\end{example}
\vspace{-5pt}

Although we generalize the results in \cite{jia2020unifying} and \cite{jia2021scalable} to structured networks with MIMO node systems, and
establish algebraic and graph-theoretic conditions for controllability of structured networks.
  However, for large-scale structured networks, both algebraic and graph theoretic criteria in Theorem \ref{th:C4SN} are computationally prohibitive since the dimension of the structured network is $n \times (n+m)$, where $n=\sum\limits_{k=1}^{N} n_k$ may be extremely large.

 A similar problem appears in the context of structured networks with SISO node systems.
 In \cite{jia2021scalable}, a scalable method was established to verify the full row rank property of the pattern matrices in Theorem \ref{th:C4SN}.
 Unfortunately, this method is not suitable for the case of structured networks with general MIMO node systems.
 Even in the circumstances of $r_k=p_k=2$, there are at least 21 classifications of node systems, and for more general circumstances, the classification is quite complicated so the scalable method is not feasible.

In the sense of the above-mentioned computational limitation, Theorem \ref{th:C4SN} is not good enough.
 Notice that a structured network is composed of two aspects, namely a family of structured node systems and structured interconnection laws.
Reasonable access to establish  computationally efficient conditions is to analyze controllability of structured networks from these two perspectives.
 Therefore, it is meaningful to provide better methods to relax the limitation of high computational cost.

  \section{NECESSARY CONDITIONS FOR CONTROLLABILITY OF LARGE-SCALE STRUCTURED NETWORKS}
  \label{sec:necessary}

 In the previous section, we have stated algebraic and graph-theoretic conditions for controllability of structured networks.
 However, these criteria are computationally prohibitive and do not take into account the characteristics of node systems and network topology, which plays an important role in structured networks.
 Therefore, in this section, to relax the computational limitation, we will provide conditions for controllability of large-scale structured networks from two perspectives, i.e. node systems and network topology.

 First, we present the following corollary from the perspective of node systems. This corollary only needs to verify node systems with dimension $n_k \times (n_k + r_k)$, which reduces computational cost compared with the whole structured network with dimension $n \times (n + m)$.
 \begin{corollary}
   If the structured network $(\mathcal{A},\mathcal{B},\mathcal{C},\mathcal{W},\mathcal{H})$ is controllable, then $(\mathcal{A}^{(k)}, \mathcal{B}^{(k)})$ is controllable for all $k=1,\ldots,N$.
 \end{corollary}
 Note that in \cite[Corollary 11]{jia2021scalable}, it has been shown that structured networks with SISO nodes are controllable only if each node system is controllable.
 The proof of the above corollary can be easily obtained by generating that from \cite[Corollary 11]{jia2021scalable}, and hence we omit it.

Next, we move on to exploring the relations between the controllability of structured networks and structured interconnection laws.
 Our first observation is that there is no direct relationship between the controllability of structured networks and the full-rank property of the pattern matrix $\begin{bmatrix} \mathcal{W} & \mathcal{H} \end{bmatrix}$  formed by structural interconnection laws.
 That is, the full-rank property of pattern matrix $\begin{bmatrix} \mathcal{W} & \mathcal{H} \end{bmatrix}$ is neither sufficient nor necessary for controllability of $(\mathcal{A},\mathcal{B},\mathcal{C},\mathcal{W},\mathcal{H})$.
On the one hand, the deficiency of sufficiency is evident since controllability of $(\mathcal{A},\mathcal{B},\mathcal{C},\mathcal{W},\mathcal{H})$ relies on not only the interconnection but also the nodal dynamics.
On the other hand, we provide the following example to illustrate that necessity does not hold either.

 \begin{example}\label{ex:3}
   Consider the structured networks $(\calA,\calB,\calC,\calW,\calH)$ in Example \ref{ex:1} which is controllable.
   We will now show that the interconnection pattern matrix $\begin{bmatrix}
  \calW & \calH
   \end{bmatrix}$ is not full row rank.
   To this end, let us consider the graph in Figure 2.
   By adopting the color change rule in Section \ref{sec:2}, we obtain that the node $6$ can not be colored, i.e., the graph is not colorable.
   This implies that the full rank property of $\begin{bmatrix} \mathcal{W} & \mathcal{H} \end{bmatrix}$ is not necessary for controllability of $(\mathcal{A},\mathcal{B},\mathcal{C},\mathcal{W},\mathcal{H})$.
\vspace{-8pt}
\begin{figure}[!h]
\centering 
\includegraphics[scale=0.8]{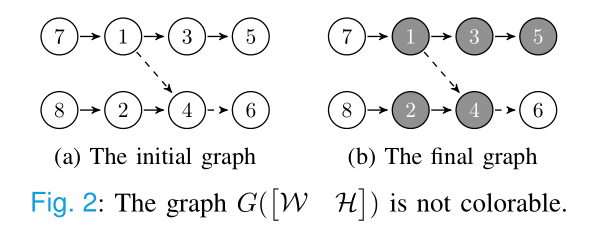}

\label{2}
\end{figure}
 \end{example}
\vspace{-8pt}

 \begin{remark}
 One of the reasons for the counter-example in Example \ref{ex:3} is that due to the inner dynamics of node $3$, the edge between vertex $24$ and vertex $32$ in Figure 1 is unnecessary to guarantee the colorability of the graph $G(\begin{bmatrix}
\calA + \calB \calW \calC & \calB \calH
 \end{bmatrix})$.
 This is unlike the case of structured networks with SISO node systems in which structured interconnection laws can accurately describe the underlying network topology.
  \end{remark}

 Fortunately, since structured interconnection laws are usually sparse, it will be particularly significant to establish conditions for controllability of structured networks with respect to the structured interconnection laws.
To this end, we introduce the following definition to extract the sparse underlying network topology.

 \begin{definition}\label{de:extract}
   Consider the interconnection pattern matrices $\mathcal{W} \in \{0,\ast,?\}^{r \times p}$ and $\mathcal{H} \in \{0,\ast,?\}^{r \times m}$.
    We will define the following underlying interconnection pattern matrices  $\widetilde{\mathcal{W}} \in \{0,\ast,?\}^{N \times N}$ and $\widetilde{\mathcal{H}} \in \{0,\ast,?\}^{N \times m}$
 \begin{equation} \label{eq:Wtilde}
     \widetilde{\mathcal{W}}_{ij} = \begin{cases}
        0 & \text{if} \quad \mathcal{W}^{(ij)} = \textbf{
        0}\\
        ? & \text{if} \quad \mathcal{W}^{(ij)} ~\text{does not contain any}~ \ast ~\text{entries and}\\
          & \text{contains at least }~?~\\
        \ast & \text{otherwise}.
     \end{cases}
 \end{equation}
 and
 \begin{equation}\label{eq:Htilde}
     \widetilde{\mathcal{H}}_{ij} = \begin{cases}
        0 & \text{if} \quad \mathcal{H}^{(ij)} = \textbf{
        0}\\
        ? & \text{if} \quad \mathcal{H}^{(ij)} ~\text{does not contain any}~ \ast ~\text{entries and}\\
          & \text{contains at least }~?~\\
        \ast & \text{otherwise}.
     \end{cases}
 \end{equation}
 \end{definition}

 Before presenting our results, we first introduce the following notions for graphs associated with pattern matrices.
 Consider a pattern matrix $\calM \in \{0,\ast,?\}^{p \times q}$ with $p < q$ and its associated graph $G(\calM)$ defined in Section \ref{sec:2}.
%
 To do so, we first need to introduce a new color change rule called {\it weakly color change rule}:
 \begin{enumerate}
 \item color all vertices in $\{p+1,\ldots,q\}$ black while the rest white;
 \item if a vertex $i$ is black and $j$ is a white out-neighbor of $i$ such that  $(i,j) \in E_{\ast}$, then change color of $j$ to black;
 \item repeat the step $2$ until no more color changes are possible.
 \end{enumerate}

 The derived set $\mathcal{D}_{w}(\calM)$ of $G(\calM)$ is defined as the set of all black nodes obtained by applying the above procedure to $G(\calM)$.
In the special case that $\mathcal{D}_{w}(\calM) = \{1,2,\ldots,q\}$, we call the graph $G(\calM)$ {\it weakly colorable}.

Now, we are ready to present the following necessary condition for controllability of large-scale structured network.

\begin{theorem} \label{th:ecto}
Consider the structured network $(\mathcal{A},\mathcal{B},\mathcal{C},\mathcal{W},\mathcal{H})$ with interconnection pattern matrices $\widetilde{\mathcal{W}}$ and $\widetilde{\mathcal{H}}$ defined as \eqref{eq:Wtilde} and \eqref{eq:Htilde}, respectively.
 Then,  the structured network $(\mathcal{A},\mathcal{B},\mathcal{C},\mathcal{W},\mathcal{H})$ is controllable only if the graph $G(\begin{bmatrix}\widetilde{\mathcal{W}}&\widetilde{\mathcal{H}}\end{bmatrix})$ is weakly colorable.
\end{theorem}
Note that Theorem \ref{th:ecto} contributes to reducing the dimension from $n \times (n + m)$ to $N \times (N + m)$, which greatly reduces computational cost.
To conclude this section, we will provide an illustrative example of the above theorem, and the proof of this theorem can be found in the appendix.
\begin{example}\label{ex:4}
   Consider the structured networks $(\calA,\calB,\calC,\calW,\calH)$ in Example \ref{ex:1} which is controllable, and  we will show that the necessary condition in the Theorem \ref{th:ecto}  holds.
   By the Definition \ref{de:extract}, we compute the following underlying pattern matrices
   \begin{equation*}
   \widetilde{\mathcal{W}} = \begin{bmatrix}
   0 & 0 & 0\\
   \ast & 0 & 0\\
   0 & \ast & 0
   \end{bmatrix} \quad \text{and} \quad\widetilde{\mathcal{H}} = \begin{bmatrix}
   \ast& \ast\\
   0 & 0\\
   0 & 0
   \end{bmatrix}.
   \end{equation*}
   It turns out that the graph $G(\begin{bmatrix}
   \widetilde{\mathcal{W}} & \widetilde{\mathcal{H}}
   \end{bmatrix})$ is weakly colorable as shown in Figure 3.
  \vspace{-7pt}

  \begin{figure}[!h]
\centering 
\includegraphics[scale=0.8]{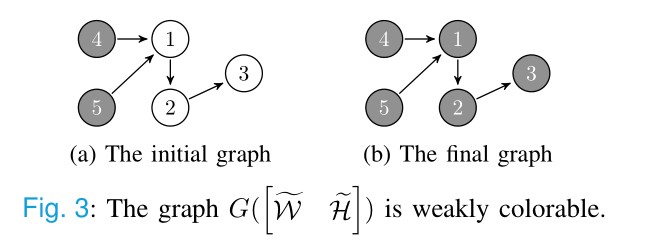}

\label{3}
\end{figure}
 \end{example}
\vspace{-15pt}


\section{Conclusion and Discussion}
\label{sec:condition}
This article has studied strong structural controllability of structured networks with general MIMO node systems. Necessary and sufficient algebraic and graph-theoretic conditions have been established which are generalizations of the results of strong structural controllability for structured systems or structured networks with SISO node systems.
The state space dimension of structured networks can be extremely large such that these criteria are computationally prohibitive.
To deal with this problem, some efficient conditions have been given from the perspective of node systems and underlying network topology extracted from structured interconnection laws. In order to verify the necessary condition from the underlying network topology, we introduce a new color change rule called weakly color change rule.

We conclude this section with some suggestions for future research.
Note that the results to deal with the problem of large-scale structured networks are all necessary conditions. Hence, finding sufficient conditions for strong structural controllability of structured networks is still an open problem.
Another opportunity for future research is to extend our results to a wider range of applications.
For example, based on the results of weak structural controllability of networks, numerous works have been reported from rather diverse perspectives on such topics as topology design \cite{mu2022structural}, minimal input selection problem \cite{Zhou2017minimal}, and so on.

\appendices
\section{The proof of Theorem \ref{th:ecto}}
For the proof of Theorem \ref{th:ecto}, the following auxiliary result will be instrumental:
\begin{lemma}\label{le:bo}
  Consider a square patter matrix $\calM \in \{0,\ast,?\}^{p \times p}$.
  Then, it holds that at most one of $\calM$ and $\calM + \calI$ has full row rank.
\end{lemma}
\begin{proof}
To prove this lemma, we will show the following statements hold: {\em (s1)} $\calM$ full row rank implies that $\calM + \mathcal{I}$ does not have full row rank, and {\em (s2)} $\calM + \mathcal{I}$ full row rank implies that $\calM$ does not have full row rank.
To prove the statement {\em (s1)}, suppose that $\calM$ has full row rank, and let $M$ be a matrix in $\calP(\calM)$.
It follows immediately that $M$ is nonsingular, and there exists a nonzero eigen-pair  $(\lambda, z)$ of $M$ such that $z^\top (M+\lambda I) = 0$.
Since $M+\lambda I \in \calP(\calM + \mathcal{I}) $ does not have full row rank which implies that $\calM + \mathcal{I}$ does not have full row rank, and we have proved {\em (s1)}.

On the other hand, suppose that $\calM + \mathcal{I}$ has full row rank.
By \cite[Lemma 19]{jia2020unifying}, there exists two permutation matrices $P_1$ and $P_2$ such that
\vspace{-5pt}
  \begin{equation*}\label{eq:permu}
    P_1 (\calM + \mathcal{I}) P_2 = \begin{scriptsize}
    \begin{bmatrix}
                          \ast & 0 & \ldots & 0  \\
                          \otimes & \ast & \ldots & 0 \\
                          \vdots & \ddots &  \ddots & \vdots  \\
                          \otimes & \ldots &  \otimes & \ast
                        \end{bmatrix}
    \end{scriptsize}
  \end{equation*}
  where the symbol $\otimes$ indicates an arbitrary entry in the set $\{0,\ast,?\}$.
  Due to the properties of permutation of matrices and the definition of pattern matrices addition operation in Table \ref{ta:results},
  in the last column of $P_1 (\calM + \mathcal{I}) P_2$, the $\ast$  must be an diagonal entry in $\calM + \mathcal{I}$.
  This implies that all the entries in the last column of $P_1 \calM P_2$ are equal to $0$.
  It then follows that $\calM$ does not have full row rank, and thus the proof is competed.
\end{proof}

Now, we can give the proof of Theorem \ref{th:ecto}.

{\it Proof of Theorem \ref{th:ecto}:}
To begin with, suppose that the structured network $(\mathcal{A},\mathcal{B},\mathcal{C},\mathcal{W},\mathcal{H})$ is controllable,
but the graph $G(\begin{bmatrix}\widetilde{\mathcal{W}}&\widetilde{\mathcal{H}}\end{bmatrix})$ is not weakly colorable.
Without loss of generality, we partition the pattern matrix $\begin{bmatrix}\widetilde{\mathcal{W}}&\widetilde{\mathcal{H}}\end{bmatrix}$ as
 \vspace{-5pt}
 \begin{equation*}
     \begin{bmatrix}
      \widetilde{\mathcal{W}}^{(11)} &  \widetilde{\mathcal{W}}^{(12)}&\widetilde{\mathcal{H}}^{(1)}\\
      \widetilde{\mathcal{W}}^{(21)} &  \widetilde{\mathcal{W}}^{(22)} &\widetilde{\mathcal{H}}^{(2)}
     \end{bmatrix}
 \end{equation*}

 where the first row and column block corresponds to black nodes in $\{1,\ldots,N\} $ and the second to remainders.

By the definition of weak color change rule, every column in $\widetilde{\mathcal{W}}^{(21)}$ and $\widetilde{\mathcal{H}}^{(2)}$ does not contain $\ast$ elements.
According to the Definition \ref{de:extract}, it follows that the corresponding
 the corresponding pattern matrix $\begin{bmatrix}\mathcal{W}&\mathcal{H}\end{bmatrix}$ can also be partitioned  as
 \begin{equation*}
     \begin{bmatrix}\calW & \mathcal{H}\end{bmatrix} = \begin{bmatrix}
      \mathcal{W}^{(11)} &  \mathcal{W}^{(12)}&\mathcal{H}^{(1)}\\
      \mathcal{W}^{(21)} &  \mathcal{W}^{(22)} &\mathcal{H}^{(2)}
     \end{bmatrix}
 \end{equation*}
 where neither $\mathcal{W}^{(21)}$ nor $\mathcal{H}^{(2)}$ contains $\ast$ element.
 Moreover, we partition the block pattern matrices $\calA$, $\calB$ and $\calC$ as
 \begin{equation*}
      \begin{scriptsize}\begin{bmatrix}
      \mathcal{A}^{(11)} &  \textbf{0} \\
      \textbf{0} &  \mathcal{A}^{(22)}
     \end{bmatrix},
      \begin{bmatrix}
      \mathcal{B}^{(11)} &  \textbf{0} \\
      \textbf{0} &  \mathcal{B}^{(22)}
     \end{bmatrix},
      \begin{bmatrix}
      \mathcal{C}^{(11)} &  \textbf{0} \\
      \textbf{0} &  \mathcal{C}^{(22)}
     \end{bmatrix}
     \end{scriptsize}
 \end{equation*}
where the first row and column block corresponds to black nodes in $\{1,\ldots,N\} $ and the second to remainders.
Therefore, recalling the definitions of multiplication and addition operations of pattern matrices, we have that
neither $\mathcal{B}^{(22)}\mathcal{W}^{(21)}\mathcal{C}^{(11)}$ nor $\calB^{(22)} \mathcal{H}^{(1)}$ contains $\ast$ element.
%
 Moreover, by the Lemma \ref{le:bo}, it follows that either $\mathcal{A}^{(22)} + \mathcal{B}^{(22)} \mathcal{W}^{(22)} \mathcal{C}^{(22)}$ or $\mathcal{A}^{(22)} + \mathcal{I} + \mathcal{B}^{(22)} \mathcal{W}^{(22)}\mathcal{C}^{(22)}$ does not have full row rank.
 This implies that either $\begin{bmatrix}\mathcal{A} + \mathcal{B}\mathcal{W}\mathcal{C} & \mathcal{B}\mathcal{H}\end{bmatrix}$ or $\begin{bmatrix}\calA +\mathcal{I} + \mathcal{B}\mathcal{W}\mathcal{C} & \mathcal{B}\mathcal{H}\end{bmatrix}$ is not full row rank, i.e., the structured network $(\calA, \calB, \calC, \calW, \calH)$ is not controllable, and thus we have reached a contradiction.
\EP

\bibliographystyle{unsrt}
\bibliography{Ref}

\begin{thebibliography}{10}

\bibitem{Trumpf2018controllability}
J.~Trumpf and H.~L. Trentelman.
\newblock Controllability and stabilizability of networks of linear systems.
\newblock {\em IEEE Transactions on Automatic Control}, 64(8):3391--3398, 2019.

\bibitem{Xiang2020Controllability}
L.~Xiang, P.~Wang, F.~Chen, and G.~Chen.
\newblock Controllability of directed networked {MIMO} systems with
  heterogeneous dynamics.
\newblock {\em IEEE Transactions on Control of Network Systems}, 7(2):807--817,
  2020.

\bibitem{ZHOU2015On}
T.~Zhou.
\newblock On the controllability and observability of networked dynamic
  systems.
\newblock {\em Automatica}, 52:63--75, 2015.

\bibitem{hao2019new}
Y.~Hao, Z.~Duan, G.~Chen, and F.~Wu.
\newblock New controllability conditions for networked, identical {LTI}
  systems.
\newblock {\em IEEE Transactions on Automatic Control}, 64(10):4223--4228,
  2019.

\bibitem{JLY2015}
Zhijian Ji, Hai Lin, and Haisheng Yu.
\newblock Protocols design and uncontrollable topologies construction for
  multi-agent networks.
\newblock {\em IEEE Transactions on Automatic Control}, 60(3):781--786, 2015.

\bibitem{lin1974structural}
C.~T. Lin.
\newblock Structural controllability.
\newblock {\em IEEE Transactions on Automatic Control}, 19(3):201--208, 1974.

\bibitem{liu2011controllability}
Y.~Y. Liu, J.~J. Slotine, and A.~L. Barab{\'a}si.
\newblock Controllability of complex networks.
\newblock {\em Nature}, 473(7346):167--173, 2011.

\bibitem{monshizadeh2014zero}
N.~Monshizadeh, S.~Zhang, and M.~K. Camlibel.
\newblock Zero forcing sets and controllability of dynamical systems defined on
  graphs.
\newblock {\em IEEE Transactions on Automatic Control}, 59(9):2562--2567, 2014.

\bibitem{commault2019structural}
C.~Commault.
\newblock Structural controllability of networks with dynamical structured
  nodes.
\newblock {\em IEEE Transactions on Automatic Control}, 65(6):2736--2742, 2020.

\bibitem{mu2022structural}
J.~Mu, J.~Wu, N.~Li, X.~Zhang, and S.~Li.
\newblock Structural controllability of networked systems with general
  heterogeneous subsystems.
\newblock {\em Asian Journal of Control}, 24(3):1321--1332, 2022.

\bibitem{Zhang2019Structural}
Y.~Zhang and T.~Zhou.
\newblock Structural controllability of an {NDS} with {LFT} parameterized
  subsystems.
\newblock {\em IEEE Transactions on Automatic Control}, 64(12):4920--4935,
  2019.

\bibitem{Mayeda1979Strong}
H.~Mayeda and T.~Yamada.
\newblock Strong structural controllability.
\newblock {\em SIAM Journal on Control and Optimization}, 17(1):123--138, 1979.

\bibitem{Chapman2013On}
A.~Chapman and M.~Mesbahi.
\newblock On strong structural controllability of networked systems: A
  constrained matching approach.
\newblock In {\em Proc. of the American Control Conference (ACC)}, pages
  6126--6131, 2013.

\bibitem{menara2017number}
T.~Menara, G.~Bianchin, M.~Innocenti, and F.~Pasqualetti.
\newblock On the number of strongly structurally controllable networks.
\newblock In {\em Proc. of the American Control Conference (ACC)}, pages
  340--345. IEEE, 2017.

\bibitem{trefois2015zero}
M.~Trefois and J.~C. Delvenne.
\newblock Zero forcing number, constrained matchings and strong structural
  controllability.
\newblock {\em Linear Algebra and its Applications}, 484:199--218, 2015.

\bibitem{jia2020unifying}
J.~Jia, H.~J.~Van Waarde, H.~L. Trentelman, and M.~K. Camlibel.
\newblock A unifying framework for strong structural controllability.
\newblock {\em IEEE Transactions on Automatic Control}, 66(1):391--398, 2021.

\bibitem{jia2021scalable}
J.~Jia, B.~M. Shali, H.~J.~Van Waarde, M.~K. Camlibel, and H.~L. Trentelman.
\newblock Scalable controllability analysis of structured networks.
\newblock {\em IEEE Transactions on Control of Network Systems}, 9(2):891--903,
  2022.

\bibitem{SWCT2021}
B.~M. Shali, H.~J. van Waarde, M.~K. Camlibel, and H.~L. Trentelman.
\newblock Properties of pattern matrices with applications to structured
  systems.
\newblock {\em IEEE Control Systems Letters}, 6:109--114, 2022.

\bibitem{Zhou2017minimal}
T.~Zhou.
\newblock Minimal inputs/outputs for a networked system.
\newblock {\em IEEE Control Systems Letters}, 1(2):298--303, 2017.

\end{thebibliography}

\end{document}